\documentclass[a4paper,UKenglish,cleveref, autoref, thm-restate]{lipics-v2021}

\pdfoutput=1 
\hideLIPIcs  

\title{Exact Subquadratic Algorithm for Many-to-Many Matching on Planar Point Sets with Integer Coordinates} 

\titlerunning{Exact Subquadratic Algorithm for Many-to-Many Matching} 

\author{Seongbin Park}{Department of Computer Science and Engineering, POSTECH, Pohang, Republic of Korea}{seongbin.park@postech.ac.kr}{https://orcid.org/0009-0000-9018-5798}{}

\author{Eunjin Oh}{Department of Computer Science and Engineering, POSTECH, Pohang, Republic of Korea \and \url{https://sites.google.com/view/eunjinoh/}}{eunjin.oh@postech.ac.kr}{https://orcid.org/0000-0003-0798-2580}{}

\authorrunning{S. Park and E. Oh} 

\Copyright{Seongbin Park and Eunjin Oh} 



\funding{Supported by Institute of Information \& Communications Technology Planning \& Evaluation (IITP) grant funded by the Korea government (MSIT) (No.RS-2024-00440239, Sublinear Scalable Algorithms for Large-Scale Data Analysis) and the National Research Foundation of Korea (NRF) grant funded by the Korea government (MSIT) (No.RS-2024-00358505).}

\nolinenumbers 

\EventEditors{John Q. Open and Joan R. Access}
\EventNoEds{2}
\EventLongTitle{42nd Conference on Very Important Topics (CVIT 2016)}
\EventShortTitle{CVIT 2016}
\EventAcronym{CVIT}
\EventYear{2016}
\EventDate{December 24--27, 2016}
\EventLocation{Little Whinging, United Kingdom}
\EventLogo{}
\SeriesVolume{42}
\ArticleNo{23}

\DeclareMathOperator{\polylog}{polylog}
\newcommand{\cand}{\ensuremath{\tilde{G}_{\textsf{cand}}}}

\begin{document}

\maketitle

\begin{abstract}
In this paper, we study the many-to-many matching problem on planar point sets with integer coordinates: 
Given two disjoint sets $R,B \subset [\Delta]^2$ with $|R|+|B|=n$, the goal is to select a set of edges between $R$ and $B$ so that every point is incident to at least one edge and the total Euclidean length is minimized.
In the general case that $R$ and $B$ are point sets in the plane, the best-known algorithm for the many-to-many matching problem takes $\tilde{O}(n^2)$ time. 
We present an exact $\tilde{O}(n^{1.5} \log \Delta)$ time algorithm for point sets in $[\Delta]^2$. 
To the best of our knowledge, this is the first subquadratic exact algorithm for planar many-to-many matching under bounded integer coordinates.
\end{abstract}

\section{Introduction}
\label{sec:introduction}
Measuring the similarity between two point sets is a fundamental problem in computational geometry and pattern recognition~\cite{alt1988congruence,unnikrishnan2005measures,veltkamp2001state}. A widely used measure is the \textit{Earth Mover's Distance} (EMD), which captures the global distribution of point sets by finding an optimal correspondence~\cite{lv2004image,mumford1991mathematical}. Formally, for two sets $R$ and $B$ of the same cardinality, the EMD is defined as the minimum cost of a perfect matching: $\min_{\sigma: R \to B} \sum_{r \in R} \|r - \sigma(r)\|$, where $\sigma$ is a bijection. While EMD is effective at preserving the holistic structure of the sets, its definition is inherently restricted to cases where $|R| = |B|$.
To accommodate sets of different cardinalities, one might consider \textit{pointwise} heuristics such as the Chamfer distance: $\sum_{r \in R} \min_{b \in B} \|r-b\| + \sum_{b \in B} \min_{r \in R} \|b-r\|$. 
It provides a simple way to handle unequal set sizes as 
it treats each connection independently and sums $|R| + |B|$ distances without any global coordination. 
However, this often fails to capture the underlying structural correspondence, as it lacks the ``assignment'' nature that makes EMD robust.

A more principled approach to handling different sized sets, instead of resorting to pointwise measures, is to generalize the matching framework of EMD itself. By relaxing the strict one-to-one matching requirement to a \textit{many-to-many} correspondence, we arrive at the \textit{minimum link measure} introduced by~\cite{DBLP:journals/acta/EiterM97}. Formally, a many-to-many matching is a set of edges $M \subseteq R \times B$ such that every point in $R \cup B$ is incident to at least one edge in $M$. The cost is the total weight $\sum_{(r,b) \in M} \|r-b\|$. 
The minimum link measure is the minimum cost of a many-to-many matching between $R$ and $B$. 
This ``setwise'' approach preserves the global structural integrity of EMD by seeking an optimal set of edges that covers both sets simultaneously. As a single edge in $M$ can satisfy the coverage requirement for points in both $R$ and $B$, the minimum link measure maintains a mutual correspondence that is more stable and geometrically intuitive than independent pointwise distances.

In this paper, we consider the problem of computing a minimum-cost many-to-many matching, also called the minimum link measure, between two point sets in Euclidean space. This problem was originally introduced by Eiter and Mannila~\cite{DBLP:journals/acta/EiterM97}. While an efficient, optimal algorithm exists for points on a line, the problem becomes significantly more difficult in higher dimensions. Even for planar point sets, the best-known exact algorithm still suffers from a quadratic bottleneck~\cite{DBLP:conf/isaac/BandyapadhyayMS21}, mirroring the challenges found in various other geometric bipartite matching problems~\cite{gattani2023robust}. To bridge this gap, we consider the case where the inputs lie on an \emph{integer grid} as an intermediate step toward the general planar case.

\subparagraph*{Our results.}
In this paper, we present the first subquadratic-time exact algorithm for the many-to-many matching problem for point sets on an integer grid. Our main result is summarized in the following theorem.

\begin{theorem}\label{thm:main}
Given two disjoint sets\footnote{With a slight modification, we can deal with the non-disjoint case without increasing the running time.} $R$ and $B$ of points in $[\Delta]^2=\{1,2,\ldots,\Delta\}\times \{1,2,\ldots,\Delta\}$, we can compute a minimum-cost many-to-many matching between $R$ and $B$ in $\tilde O(n^{1.5}\log\Delta)$ time,\footnote{The $\tilde O$-notation hides polylogarithmic factors in $n$.} where $n=|R|+|B|$.
\end{theorem}

A key technical challenge in achieving this bound is efficiently handling subproblems where points can be either matched to the opposite set or remain unmatched by paying a certain cost. To this end, we introduce the \textit{minimum-cost matching with penalties} problem: Given two point sets $R$ and $B$ where each point $p \in R \cup B$ is associated with a real-valued penalty $\pi(p) \ge 0$, we seek a matching $M \subseteq R \times B$ that minimizes
$\sum_{(r,b) \in M} \|r-b\| + \sum_{p \in \text{free}(M)} \pi(p)$, 
where $\text{free}(M)$ denotes the set of points in $R \cup B$ not incident to any edge in $M$. 
To the best of our knowledge, this problem has not been explicitly studied before. Although this can be solved in $O(n^3)$ time via standard reductions to the minimum-cost perfect matching problem for bipartite graphs, the cubic-time algorithm is too slow for our purpose. 
As a subroutine for our main algorithm, we provide an optimal algorithm for this penalty variant in one dimension, which we believe is of independent interest, stated as follows.

\begin{theorem}\label{thm:matching-penalty}
Given two disjoint sets $R$ and $B$ of points on a real line, where each point has a real-valued penalty, a minimum-cost matching with penalties can be computed in $O(n\log n)$ time, where $n=|R|+|B|$.
\end{theorem}

\subparagraph*{Related works.}
The study of the minimum-cost many-to-many matching problem was introduced by Eiter and Mannila~\cite{DBLP:journals/acta/EiterM97}. They established that the problem is solvable in $O(n^3)$ time by reducing it to the minimum-cost perfect matching problem on graphs~\cite{DBLP:journals/acta/EiterM97}.  For the one-dimensional case where points lie on a real line, Colannino et al.~\cite{DBLP:journals/gc/ColanninoDHLMRST07} provided an optimal $O(n \log n)$-time algorithm, improving upon previous $O(n^2)$ results~\cite{colannino2005faster}. In the Euclidean plane, the best-known exact algorithm runs in $\tilde{O}(n^2)$ time~\cite{DBLP:conf/isaac/BandyapadhyayMS21}.
On the other hand, in higher dimensions, no non-trivial algorithm is known for this problem while an $O(n^3)$-time algorithm can be obtained easily by reducing it to a minimum-cost perfect matching of a graph. Given the challenge of computing exact solutions in subquadratic time, several approximation algorithms have been explored. The Chamfer distance serves as a simple $2$-approximation for the many-to-many matching cost~\cite{DBLP:conf/isaac/BandyapadhyayMS21}. Furthermore, $(1+\varepsilon)$-approximate solutions can be computed in $(1/\varepsilon)^{O(d)} \cdot n \log n$ time for $d$-dimensional Euclidean space, or any metric space with a constant doubling dimension~\cite{an2025approximation,bandyapadhyay2024n}. 

As EMD is a more classical setting, there are numerous results on computing the EMD between two points (also known as the minimum-cost perfect matching). 
By implementing the Hungarian algorithm~\cite{kuhn1955hungarian} using a dynamic weighted nearest neighbor data structure~\cite{DBLP:journals/dcg/KaplanMRSS20}, 
one can compute the exact EMD of two point sets in the plane in $\tilde O(n^2)$ time~\cite{gattani2023robust}. 
While the exact bipartite matching problem in the plane suffers from a quadratic bottleneck,
the problem becomes easier if there are constraints on the input points. 
For instance, 
if the points in $B$ and $R$ are drawn independently
and identically from a fixed distribution that is not known to the algorithm, the exact EMD between $B$ and $R$ can be computed in 
$O(n^{7/4} \log \Phi)$ time, where $\Phi$ denotes the spread of $B\cup R$ (i.e, the ratio of the maximum distance to the minimum distance).
If $B$ and $R$ come from $[\Delta]^2$, Sharathkumar \cite{DBLP:conf/compgeom/Sharathkumar13} presented an $\tilde O(n^{1.5} \log\Delta)$-time algorithm.\footnote{The paper~\cite{DBLP:conf/compgeom/Sharathkumar13} states that their algorithm takes $O(n^{1.5+\delta} \log (n\Delta))$ time for a constant $\delta$, but by replacing a nearest neighbor data structure used in this paper with the most recent one~\cite{DBLP:journals/dcg/KaplanMRSS20}, the factor $n^\delta$ can be replaced with $\polylog n$.}

\section{Preliminaries and Algorithm Overview}
In this section, we first describe an alternative view for the problem and then provide an overview of our algorithm. 
An alternative way is to view the many-to-many matching problem as a geometric graph problem. Consider a complete bipartite graph $G = ((R,B), E)$ with vertex set $V = R \cup B$ and with edge set $E = R \times B$ where 
the cost of an edge $(r, b) \in E$ is the Euclidean distance $\|r-b\|$ between their endpoints. 
An \emph{edge cover} of $G$ is a subset of $E$ such that every vertex of $V$ is incident to at least one edge of it. 
The cost of an edge cover of $G$ is defined as the sum of the costs of its edges.
Then the many-to-many matching problem on $R$ and $B$ is equivalent to the problem of finding an edge cover of $G$ with the minimum cost.

For a general graph with $n$ vertices and $m$ edges, a minimum-cost edge cover can be solved in $O(n(n\log n +m))$ time via a reduction to the minimum-cost perfect matching~\cite{DBLP:journals/jacm/FredmanT87,  DBLP:journals/talg/Gabow17,DBLP:journals/jct/KeijsperP98}. 
Since the complete bipartite graph $G$ has complexity $\Theta(n^2)$,
we cannot use these graph algorithms directly to obtain a subquadratic-time algorithm. Nevertheless, as our approach extensively adapts several combinatorial principles from these classical algorithms, 
we introduce the following graph-theoretic terminology to be used throughout this paper.

\subparagraph*{Terminology for matching.}
Let $H$ be a bipartite graph with edge cost. We use $V(H)$ to denote the vertex set of $H$. 
A \emph{matching} $M$ in $H$ is a set of pairwise vertex-disjoint edges of $H$. We use $V(M)$ to denote the set of endpoints of the edges of $M$.
A matching in $H$ is \emph{perfect} if every vertex of $H$ is incident to exactly one edge of it.
The cost of a matching is defined as the sum of the costs of its edges.
A vertex of the graph is \emph{free} (with respect to $M$) if no edge of $M$ is incident to it. 
An \emph{alternating path} (with respect to $M$) is a path in the graph that begins from a free vertex and whose edges alternately belong to $M$ and not to $M$. 
An \emph{augmenting path} (with respect to $M$) is an alternating path whose two endpoints are both free. 
Given an augmenting path $P$, we \emph{augment} $M$ along $P$ by taking the symmetric difference of $M$ and $P$, i.e., $M \oplus P$.
As a result, the cardinality of the matching increases by one. 
It is known that $M$ is a maximum-cardinality matching of $G$
if no augmenting path with respect to $M$ exists.
A classical way to compute a maximum-cardinality (or minimum-cost) matching of a graph is to repeatedly compute an augmenting path with certain properties and augment the matching until no augmenting path remains.

\subsection{Reduction to the Min-Cost Perfect Matching}
\label{sec:reduction-to-mcpm}

\begin{figure}
    \centering
    \includegraphics[width=0.8\linewidth]{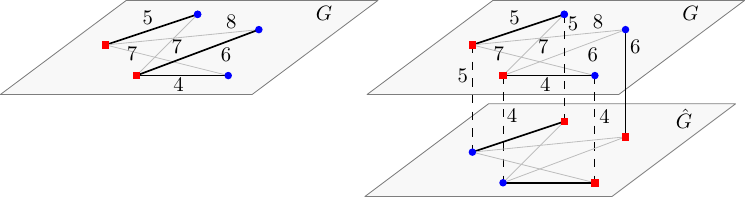}
    \caption{Illustration of the prism graph (right) of the original graph (left).
    The thick segments in the left figure show the min-cost many-to-many matching, and the thick segments in the right figure show the min-cost perfect matching.}
    \label{fig:prism}
\end{figure}

In this section, we describe a reduction from the minimum-cost many-to-many matching problem to the minimum-cost perfect matching problem.
While this is identical to the folklore reduction, we refer to the resulting graph as the \emph{prism graph} for notational convenience.
It preserves the original geometric embedding and induces independent vertex penalties.

Let $R$ and $B$ be two point sets in $[\Delta]^2$, and let $G$ be the complete bipartite graph between $R$ and $B$, where the cost of each edge is the Euclidean distance between its endpoints. 
The reduced instance is called the \emph{prism graph} of $G$, denoted by $\tilde{G} = (\tilde{V}, \tilde{E})$, with an associated edge cost $\tilde{c}: \tilde{E} \to \mathbb{R}$.

\subparagraph*{Prism graph.}
The prism graph $\tilde{G}$ is a two-layered bipartite graph consisting of an ``upper'' layer, which is the original graph $G$, and a ``lower'' layer, which is a mirrored copy $\hat{G}$ of $G$. 
For each vertex $v$ in the upper layer, we add a \textit{link edge} between $v$ and its corresponding copy $\hat{v}$ in the lower layer. 
Note that $\tilde{G}$ remains bipartite, as its vertex set can be partitioned into $R \cup \{\hat v \mid \hat v \text{ is the copy of } v \in B\}$ and $B \cup \{\hat v \mid \hat v \text{ is the copy of } v \in R\}$.
See Figure~\ref{fig:prism}.
For an edge $e=(u,v)$ in $\tilde{G}$, the cost $\tilde{c}(e)$ is defined as follows:
\begin{itemize}
    \item If $u$ and $v$ are both in the upper layer, $\tilde{c}(e) = \|u - v\|$.
    \item If $u$ and $v$ are both in the lower layer, $\tilde{c}(e) = 0$.
    \item If $e$ is a link edge $(v, \hat{v})$, $\tilde{c}(e) = \mu(v)$, where $\mu(v)$ is the distance from $v$ to its nearest neighbor in the opposite set ($R$ or $B$).
\end{itemize}

Note that an optimal edge cover $M$ of $G$ consists of disjoint stars. For a star of $M$ of size larger than two, let $c$ be the center of the star. Then for every vertex of the star other than $c$, its closest neighbor in the opposite set is $c$. Otherwise, we can reconnect it with its closest neighbor without violating the feasibility of the solution. 
This reduces the problem into the minimum-cost matching problem with penalties where the penalty of each point is $\mu(\cdot)$. 
The construction of the prism graph is based on this intuition. 
Matching a vertex via a link edge to the lower layer simulates leaving that vertex uncovered by the primary matching in the upper layer, instead paying the penalty of connecting it to its nearest neighbor.

The values $\mu(v)$ for all $v \in R \cup B$ can be computed in $O(n \log n)$ time by constructing the Voronoi diagrams of $R$ and $B$ and using a planar point location data structure~\cite{DBLP:journals/siamcomp/Kirkpatrick83}. 
Although $\tilde{G}$ has $\Theta(n^2)$ edges in the worst case, we can store it using complexity of $O(n)$ by maintaining only the vertices and the link edges explicitly. The edges in $G$ and $\hat{G}$ can be represented implicitly as they form complete bipartite graphs. 

\medskip 
The following lemma establishes the equivalence between the minimum-cost edge cover of $G$ and the minimum-cost perfect matching of $\tilde{G}$. For illustration, see Figure~\ref{fig:prism}.  

\begin{restatable}[\cite{DBLP:conf/isaac/BandyapadhyayMS21}]{lemma}{reduction}
\label{lem:mcpm-to-mcec}
Given a minimum-cost perfect matching $\tilde{M}^*$ of $\tilde{G}$, a minimum-cost edge cover of $G$ can be computed in $O(n)$ time.
\end{restatable}
\begin{proof}
Let $V=(R,B)$ and $E$ denote the vertex set and edge set of the upper layer of $\tilde G$, respectively.  
Bandyapadhyay et al. showed that for any edge cover of $G$ with cost $D$, there exists a perfect matching in $\tilde{G}$ with cost at most $D$. 
Conversely, for any perfect matching in $\tilde{G}$ with cost $D$, there exists an edge cover of $G$ with cost $D$.
Consequently, the cost of a minimum-cost edge cover of $G$ is equal to the cost of a minimum-cost perfect matching of $\tilde{G}$.
By slightly modifying its proof, we can prove the lemma stated above.

We construct a set of edges $C^*$ in $G$ from $\tilde{M}^*$ as follows.
First, we include all edges in $\tilde{M}^* \cap E$ into $C^*$.
Next, for every vertex $v \in R \cup B$ that is not incident to any edge in $\tilde{M}^* \cap E$, we add the edge connecting $v$ to its nearest neighbor in the opposite set, i.e., in $B$ if $v \in R$, and in $R$ if $v \in B$, to $C^*$.
Since the nearest neighbors can be pre-computed, this construction takes $O(n)$ time.

We now show that the cost of $C^*$ equals the cost of $\tilde{M}^*$.
If a vertex $v$ is not covered by $\tilde{M}^* \cap E$, it must be matched in $\tilde{M}^*$ via a link edge.
By the definition of the cost function $\tilde{c}$, the cost of this link edge is exactly $\mu(v)$, which corresponds to the cost of the edge connecting $v$ to its nearest neighbor in $G$.
Therefore, the total cost is preserved.
Bandyapadhyay et al. showed that $C^*$ is a minimum-cost edge cover of $G$.
\end{proof}

\subsection{Overview of Our Algorithm}
Due to Lemma~\ref{lem:mcpm-to-mcec}, the problem reduces to computing a minimum-cost perfect matching of $\tilde G$.
The algorithm consists of two phases.

In the first phase, we compute an approximate perfect matching of $\tilde{G}$ with its corresponding dual weights. We use the approximation algorithm of Bandyapadhyay et al.~\cite{DBLP:conf/isaac/BandyapadhyayMS21} which is based on the scaling algorithm~\cite{DBLP:journals/siamcomp/GabowT89}, with dynamic additively weighted nearest-neighbor data structures~\cite{DBLP:journals/dcg/KaplanMRSS20} to obtain these dual weights in $\tilde{O}(n^{1.5} \log \Delta)$ time. We provide the details of this step in Appendix~\ref{sec:scaling-algo-for-1-opt} to make the paper self-contained.

In the second phase, we leverage these dual weights to construct a sparse candidate subgraph $\tilde{G}_{\textsf{cand}}$ of $\tilde{G}$ containing an optimal perfect matching of $\tilde G$.
Here, a crucial property is that it is \emph{almost} planar. Specifically, the upper layer of $\tilde{G}_{\textsf{cand}}$ possesses a planar \emph{skeleton}, meaning no two edges of $\cand$ cross unless their endpoints are collinear, and the lower layer of $\cand$ is a mirrored copy of its upper layer. We then compute the exact minimum-cost perfect matching of $\tilde{G}_{\textsf{cand}}$ by applying a divide-and-conquer strategy on a balanced separator of this planar skeleton.

\medskip 
While this overall strategy comes from the subquadratic-time algorithm for the one-to-one setting on integer grids~\cite{DBLP:conf/compgeom/Sharathkumar13}, we address two major technical hurdles unique to the many-to-many setting.
\begin{itemize}
    \item During the construction of $\tilde{G}_{\textsf{cand}}$ and within the recursion step, we must efficiently solve subproblems where points can either be matched to the opposite set or left unmatched by paying a penalty. To handle this, we develop an $O(n \log n)$-time algorithm for the 1D minimum-cost matching problem with penalties (Section~\ref{sec:1D}), replacing the simple greedy approach used for standard 1D perfect matching.
    \item In the one-to-one setting~\cite{DBLP:conf/compgeom/Sharathkumar13}, the candidate graph $\cand$ is strictly planar, allowing direct application of the separator-based algorithm  in~\cite{DBLP:journals/siamcomp/LiptonT80}. In our case, $\tilde{G}_{\textsf{cand}}$ is not strictly planar. We overcome this by applying the balanced separator exclusively to the planar skeleton of $\tilde{G}_{\textsf{cand}}$ (Section~\ref{sec:exact-algo-for-mcpm:dnc-algo}). 
    A tricky part is to handle the vertices of $\cand$ not appearing on the planar skeleton. We address this using their 1D nature: these vertices are fully contained in the open intervals (edges) of the planar skeleton, which allows us to process them efficiently using our 1D matching subroutine.
\end{itemize}

\section{Subroutine: 1D Minimum Cost Bipartite Matching with Penalties}\label{sec:1D}
In this section, we prove Theorem~\ref{thm:matching-penalty}, that is, we present an $O(n\log n)$-time algorithm for the 1D minimum-cost bipartite matching problem with penalties. As input, we are given two sets $B$ and $R$ lying on a horizontal line where each point $p\in B\cup R$ is associated with a penalty $w_p$.
For a matching $M\subseteq B\times R$ between $B$ and $R$, its cost $c(M)$ is defined as $\sum_{(r,b)\in M}\|r-b\| + \sum_{p\in (R\cup B)\setminus V(M)} w_p$. That is, its cost is the total edge length of $M$ plus the total penalty of the unmatched vertices. The goal is to compute a matching between $R$ and $B$ of minimum total cost.

To the best of our knowledge, there is no work explicitly studying the minimum-cost bipartite matching problem with penalties even for the one-dimensional case. 
However, there are several ways to address this problem.
First, this problem reduces to the min-cost flow problem on planar directed graphs. 
Also, if points of $B\cup R$ have integer coordinates, and the penalties $w_p$ are all integers, then the resulting instance has integer costs. In this case, we can solve the problem in $\tilde O(n\log N)$ time, where $N$ is the maximum of the maximum cost and the diameter of $B\cup R$. However, the penalty $w_p$ can be an arbitrary real number in our case, which makes the aforementioned algorithm inapplicable in our setting.
Another simple way is to use dynamic programming, which leads to an $O(n^2)$-time algorithm. Since our main algorithm is based on this dynamic programming algorithm, we describe it in more detail here.

\subsection{Quadratic-Time DP Algorithm}
Let $\ell$ be the line containing $B$ and $R$. 
Without loss of generality, assume that $\ell$ is horizontal.
For a point $x\in \ell$, let $B_x$ and $R_x$ be the sets of points in $B$ and $R$, respectively, lying to the left of $x$ (including the points lying on $x$).
For any point $x\in \ell$ and a positive integer $k$, we let $F_x(k)$ denote the minimum matching cost with penalties between $\bar B_x$ and $R_x$,
where $\bar B_x$ is the set obtained from $B_x$ by adding $k$  points at $x$ with penalty $\infty$.
For a negative integer $k$, we let
$F_x(k)$ denote the minimum matching cost with penalties between $B_x$ and $\bar R_x$,
where $\bar R_x$ is the set obtained from $R_x$ by adding $|k|$  points at $x$ with penalty $\infty$.
Then the minimum matching cost between $B$ and $R$ is $F_\infty(0)$. Therefore, the problem reduces to computing $F_x(\cdot)$ for all points $x$ in $\ell$ and for all values $k$ with $-n\leq k \leq n$.

\begin{lemma}
    For any value $k$ and any two points $x$ and $x'$ of $\ell$ such that no point of $R\cup B$ lies between $x$ and $x'$ (including $x$ and $x'$),
    we have $F_x(k)=F_{x'}(k)+|k|\cdot \|x-x'\|$.
\end{lemma}
\begin{proof}
    Without loss of generality, assume that $x$ lies to the left of $x'$. Note that $B_x = B_{x'}$ and $R_x = R_{x'}$. We consider the case $k > 0$; the case for $k \leq 0$ follows by symmetric arguments. 
    
    Let $M_{x'}$ be an optimal matching between $\bar{B}_{x'}$ and $R_{x'}$ attaining $F_{x'}(k)$. By definition, $\bar{B}_{x'}$ consists of $B_{x'}$ and $k$ virtual points located at $x'$ with penalty $\infty$. Because the penalty is infinite, any matching with finite cost must match all $k$ virtual points to $k$ distinct points in $R_{x'}$.
    Since no point of $B\cup R$ lies between $x$ and $x'$, we have $R_x = R_{x'}$, and every $r \in R_{x'}$ lies to the left of or at $x'$. Hence, $
\|r - x'\| = \|r - x\| + \|x - x'\|$. 
 If we move the $k$ virtual points from $x'$ to $x$, the cost of each of the $k$ edges matched to these virtual points decreases by exactly $\|x - x'\|$. The matching status and penalties of all other points in $(B_{x'} \cup R_{x'}) \setminus V(M_{x'})$ remain unchanged. Thus, we have 
    $F_x(k) \leq F_{x'}(k) - |k| \cdot \|x - x'\|$. 
    
    Conversely, let $M_x$ be an optimal matching attaining $F_x(k)$. By shifting the $k$ virtual points from $x$ to $x'$, the distance from any $r \in R_x$ to the virtual points increases by exactly $\|x - x'\|$. This construction yields 
    $F_{x'}(k) \leq F_x(k) + k \cdot \|x - x'\|$.   
    Combining these two inequalities, we obtain $F_x(k) = F_{x'}(k) + |k| \cdot \|x - x'\|$. For $k < 0$, the same logic applies to the $|k|$ virtual points added to $R$, resulting in $F_x(k) = F_{x'}(k) + |k| \cdot \|x - x'\|$.
\end{proof}

Thus it is sufficient to compute $F_p(k)$ for all points $p\in R\cup B$ and all values $-n\leq k\leq n$. 
But to avoid a degeneracy issue, for each point $p\in R\cup B$, 
we compute $F_{\bar p}(k)$ for a conceptual point $\bar p$ located infinitesimally to the right of $p$. For all distance calculations, we treat $\|p-\bar p\|=0$.
As initialization, let $x_0$ be any point lying strictly to the left of the first point in $B \cup R$. 
Then $F_{x_0}(0)=0$ and $F_{x_0}(k)=\infty$ for any $k\neq 0$.
Starting from this base case, we process the points of $R\cup B$ from left to right along $\ell$.  
We show how to handle a point $p\in R\cup B$.
Without loss of generality, assume that $p\in R$. The other case can be handled symmetrically. Let $p'$ be the point of $B\cup R$ lying immediately to the left of $p$ along $\ell$, and let $d = \|p - p'\|$.  
Then we have the following recurrence relation. 
There are two  possibilities: either $p$ is matched with a blue point in $\bar B_{\bar p}$, or $p$ is not used in the matching by paying the penalty. We consider both cases, and take the one with minimum cost. 
\begin{equation}\label{eqn:dp}
        F_{\bar p}(k) = \min \Big( \underbrace{F_{\bar p'}(k-1) + |k-1|d}_{\text{Match } p}, \quad \underbrace{F_{\bar p'}(k) + |k|d + w_p}_{\text{Pay Penalty}} \Big)
    \end{equation}

The recurrence relation immediately shows how to compute $F_{\bar p}(\cdot)$ for all points $p\in B\cup R$. Each table entry can be computed in constant time, yielding an overall running time of $O(n^2)$.

\medskip

In the following ,
to make the description easier, we simply let $F_p(k):=F_{\bar p}(k)$
as the following analysis is based only on the recurrence relation~(\ref{eqn:dp}).

\subsection{Convexity of the Cost Function}
In this subsection, we show that the function $F_x(\cdot)$ is \emph{convex} for any fixed point $x\in \ell$. 
We say $f: \mathbb{Z}\rightarrow \mathbb{R}$ is \emph{convex} if 
$f(i-1)+f(i+1)\geq 2f(i)$ for every integer $i$.
If $f(i+1)=\infty$ or $f(i-1)=\infty$, the inequality holds immediately.
We let $\textrm{dom}(f)$ be the set of integers $k$ such that $f(k)$ is finite. Let $\Delta f(k) :=f(k)-f(k-1)$.

\begin{lemma}
    The function $F_x(\cdot)$ is convex for any fixed point $x\in R\cup B$.
\end{lemma}
\begin{proof}
    We use induction on the points of $B \cup R$ encountered from left to right.
    Let $x_0$ be a point to the left of the leftmost point of $B \cup R$. We define $F_{x_0}(0) = 0$ and $F_{x_0}(k) = \infty$ for $k \neq 0$, which is trivially convex.

Let $p \in B \cup R$ be the current point and $p'$ be the previous point in $B \cup R \cup \{x_0\}$. Let $d = \|p - p'\|$. 
Without loss of generality, assume that $p\in B$.
By the induction hypothesis, $F_{p'}(\cdot)$ is convex. 
    Note that $|k-1|d$ is a convex function of $k$. Since the sum of two convex functions is also convex, $F_{p'}(k-1)+|k-1|d$ is convex.
    Similarly, $F_{p'}(k)+|k|d+w_p$ is convex. 
    Although the minimum of two convex functions is not convex in general, it is convex in our setting.
    This holds due to the following claim with $f(k)=F_{p'}(k)+|k|d$.
    \begin{claim}\label{claim:convex}
        Let $f: \mathbb{Z} \rightarrow \mathbb{R}$ be a convex function. Then $F_p(k) := \min\{f(k-1), f(k) + w\}$ is also convex for any value $w$.
    \end{claim}
    \begin{proof}
    To prove convexity, it suffices to show that $\Delta F_p(k) \le \Delta F_p(k+1)$ for all integers $k$. 
Let $k^*$ be the largest integer satisfying
$\Delta f(k^*)= f(k^*) - f(k^*-1) \le -w$. 
Since $f$ is convex, we have $\Delta f(k) \le \Delta f(k+1)$ for all $k \in \mathrm{dom}(f)$.
Therefore,
for $k\leq k^*$, we have $F_p(k)=f(k)+w$, and for $k>k^*$, we have $F_p(k)=f(k-1)$.  

We now show that $F_p(\cdot)$ is convex. 
Away from $k^*$ and $k^*+1$, the convexity follows immediately from the previous observation. 
At $k^*$, we have
\begin{align*}
\Delta F_p(k^*) - \Delta F_p(k^*+1)
&\leq f(k^*) - f(k^*-1)- (f(k^*) - f(k^*) - w) \\
&= \Delta f(k^*) + w \\
&\le 0.
\end{align*}
Similarly, at $k^*+1$, we have
\begin{align*}
\Delta F_p(k^*+1)-\Delta F_p(k^*+2) 
&\leq (f(k^*) - f(k^*)-w)- (f(k^*+1) - f(k^*)) \\
&= -\Delta f(k^*+1) -w\\
&\le 0.
\end{align*}
The last inequality holds since $k^*$ is the largest integer with $\Delta f(k^*) \le -w$, and $f$ is convex.
Therefore, at any integer $k$, the function $F_p(k)$ is also convex, and thus the claim holds.
\end{proof}

Therefore, the lemma holds for any fixed point $x\in R\cup B$.
\end{proof}

As we showed in the proof of Claim~\ref{claim:convex}, the following corollary holds.
\begin{corollary}\label{cor:search}
Let $k^*$ be the largest integer with 
$\Delta f(k) \le -w_p$.
    Then for $k\leq k^*$, we have $F_p(k)=f(k)+w_p$, and for $k>k^*$, we have $F_p(k)=f(k-1)$. 
\end{corollary}

\subsection{Near-Linear-Time Algorithm}
\label{sec:near-linear-algo}
In this subsection, we present an $O(n\log n)$-time algorithm for the 1D minimum-cost bipartite matching problem with penalties. 
For each $p\in B\cup R$, let $\mathrm{dom}(F_p)=[L_p,R_p]$ be the maximal interval
of integers for which $F_p(k)<\infty$.
By the recurrence, $\mathrm{dom}(F_x)$ is always a contiguous interval.
Basically, we compute $F_p(k)$ for all integers $k\in\textrm{dom}(F_p)$ and all points $p\in B\cup R$. 
By maintaining $F_p(k)$ using a data structure for a fixed $p$, we show that $F_p(k)$ can be computed in $O(\log n)$ time for all integers $k$ for a fixed $p$ assuming that we have $F_{p'}(\cdot)$, where $p'$ is the point of $B\cup R$ 
lying immediately to the left of $p$.

\subparagraph{Representation of the cost function.}
Fix $x\in\ell$, and let $\Delta F_x(k) := F_x(k)-F_x(k-1)$. 
We maintain $F_x(0)$ and $\Delta F_x(k)$ for all indices $k$. Note that they fully represent $F_x(\cdot)$.
To make the update efficiently, we store $\Delta F_x(\cdot)$ using two balanced binary search trees.
Let $H^+_x$ and $H^-_x$ be two sets where 
\[H^+_x=\{\Delta F_x(k) \mid 0<k\le R_x\} \text{ and } H^-_x=\{\Delta F_x(k)\mid L_x< k\le 0\}.\]

Since $F_x(\cdot)$ is convex, $\Delta F_x(k)$ is non-decreasing with respect to $k$.
We maintain the two sets using two binary search trees $T^+_x$ and $T^-_x$ and offsets $o^+$ and $o^-$. 
The offsets allow us to shift all stored values uniformly without updating each element explicitly. 
Each leaf of a tree of $T_x^+$ (and $T_x^-$) corresponds to exactly one index $k>0$ (and $k<0$), and the leaf stores the value 
$\Delta F_x(k) - o^+$ (and $\Delta F_x(k) - o^-$).
Thus, the actual value of $\Delta F_x(k)$ is obtained by adding the corresponding offset $o^+$ or $o^-$. 
Each internal node additionally stores the size of its subtree and the sum of the stored values in its subtree (excluding the offset). 
Hence, the true sum of the values in any subtree can be obtained by adding the offset multiplied by the subtree size. 
The leaves of $T_x^+$ (and $T_x^-$) are stored in increasing order of their indices $k$.
Since $F_x(\cdot)$ is convex, the stored values are non-decreasing along this order.
Therefore, $T_x^+$ and $T_x^-$ can be implemented as binary search trees augmented with subtree sizes and subtree sums, supporting insertion, deletion, and prefix-sum queries in $O(\log n)$ time.
We call the tuple of $F_x(0)$, $o^+$, $o^-$, $T_x^+$ and $T_x^-$ the \emph{representation} of $F_x(\cdot)$. 
Here, since the roots store the size of the trees, we do not need to explicitly maintain $\textrm{dom}(F_p)$.  

\begin{observation}\label{obs:value}
    Given the representation of $F_x(\cdot)$, we can compute $F_x(k)$ for any integer $k$ in $O(\log n)$ time.
\end{observation}
\begin{proof}

We show how to compute $F_x(k)$ for any integer $k>0$. For $k=0$, we maintain $F_x(0)$ explicitly, and for $k<0$, we can do this symmetrically. 
By telescoping, we have
\[
F_x(k)
= F_x(0) + \sum_{i=1}^{k} \Delta F_x(i).
\]
Thus, it suffices to compute the prefix sum of the first $k$ values in $H_x^+$.

Since $\Delta F_x(k)$ is non-decreasing in $k$, the values stored in $T_x^+$ are already ordered by index.
We augment each node of $T_x^+$ with the size of its subtree and the sum of the stored values in its subtree (excluding the offset $o^+$).
Using standard binary search tree operations, we can retrieve the sum of the smallest $k$ elements in $T_x^+$ in $O(\log n)$ time.
Adding this sum and the accumulated offset $k \cdot o^+$ to $F_x(0)$ yields $F_x(k)$. 
Therefore, given the representation of $F_x(\cdot)$, we can compute $F_x(k)$ for any integer $k$ in $O(\log n)$ time.
\end{proof}

\subparagraph{Update of the representation.}

Assume that we are given the representation of $F_{p'}(\cdot)$, and we wish to compute the representation of $F_p(\cdot)$, where $p$ is the next point to the right of $p'$ on $\ell$. 
We describe the case $p \in B$; the case $p \in R$ is symmetric.
Let $d := \|p-p'\|$.

\medskip 
\noindent\textbf{Step 1.}
We first consider the transportation cost. Let $f(k) := F_{p'}(k) + |k|d$. 
Thus, before handling the matching/penalty decision at $p$, we first transform the representation of $F_{p'}$ into the representation of $f$. For $k>0$, we have
$\Delta f(k)= \Delta F_{p'}(k)+d$, and 
for $k\leq 0$, we have $\Delta f(k)= \Delta F_{p'}(k)-d$.
Hence, this update can be performed by
 increasing the offset $o^+$ by $d$, and
 decreasing the offset $o^-$ by $d$, 
while keeping the trees unchanged.
Also, $F_{p'}(0)$ remains unchanged since $|0|d=0$.
After this step, the representation corresponds to $f(\cdot)$.

\medskip 
\noindent\textbf{Step 2.}
We then find the largest integer $k^*$ with 
$\Delta f(k) \le -w_p$
by a binary search over the ordered trees $T_{p'}^+$ and $T_{p'}^-$ in $O(\log n)$ time.
Since $\Delta f(k)$ is non-decreasing in $k$, 
the threshold $k^*$ can be found by searching for the largest index 
whose discrete derivative is at most $-w_p$. 
Recall that 
$F_p(k) = \min \{ f(k-1),\, f(k) + w_p \}$ from~\eqref{eqn:dp}. 
By Corollary~\ref{cor:search}, 
for $k\leq k^*$, we have $F_p(k)=f(k)+w_p$, and for $k>k^*$, we have $F_p(k)=f(k-1)$. 

\medskip
\noindent\textbf{Step 3.} 
We now update the representation of $f(\cdot)$ to obtain that of $F_p(\cdot)$.
We first compute $F_p(0)$.
Since $F_p(0)=\min\{f(-1),\,f(0)+w_p\}$, 
this value can be computed in $O(\log n)$ time using
Observation~\ref{obs:value}.

We next determine the new discrete derivatives.
For $k \le k^*$, both $F_p(k)$ and $F_p(k-1)$ lie in the same case, and hence
\[
\Delta F_p(k)=F_p(k)-F_p(k-1)
             =f(k)-f(k-1)
             =\Delta f(k).
\]
For $k > k^*+1$, we have
\[
\Delta F_p(k)
=F_p(k)-F_p(k-1)
=f(k-1)-f(k-2)
=\Delta f(k-1).
\]
At the boundary $k=k^*+1$, we have
\[
\Delta F_p(k^*+1)
=F_p(k^*+1)-F_p(k^*)
=f(k^*)-(f(k^*)+w_p)
=-w_p.
\] 
Therefore, the sequence $\Delta F_p(\cdot)$ is obtained from
$\Delta f(\cdot)$ by inserting the single value $-w_p$ (minus the corresponding offset) 
at position $k^*+1$.
As a result, all indices strictly larger than $k^*$ are shifted by one position to the right.

Since the discrete derivatives are stored in a binary search tree,
this transformation can be implemented by a single rank-based insertion
at index $k^*+1$.
If $k^* < 0$, the insertion shifts the element originally at index $0$ to index $1$. 
Because index $0$ belongs to $T_{p'}^-$ and index $1$ belongs to $T_{p'}^+$, we must extract the maximum element from $T_{p'}^-$ and insert it as the minimum element in $T_{p'}^+$ (after adjusting the values with respect to the offsets for the two trees).
All subtree sizes and subtree sums are updated along the search path. 
The insertion and any necessary boundary shifts take $O(\log n)$ time.
Thus, the entire update from the representation of $F_{p'}(\cdot)$
to that of $F_p(\cdot)$ takes $O(\log n)$ time.

\begin{lemma}
    Given the representation of $F_{p'}(\cdot)$, 
    we can compute the representation of $F_p(\cdot)$ in $O(\log n)$ time.
\end{lemma}

We consider the points of $B\cup R$ from left to right. The initialization takes $O(n)$ time, and each update takes $O(\log n)$ time. Therefore, the total running time is $O(n\log n)$.

\section{Minimum-Cost Perfect Matching for Prism Graphs}
\label{sec:exact-algo-for-mcpm}
In this section, we present an exact algorithm for computing a minimum-cost perfect matching of the prism graph $\tilde{G}$.
To simplify the exposition, we identify the vertices in the \emph{upper layer} of the prism graph $\tilde{G}$ directly with the original points in $R$ and $B$. 
Specifically, we let the upper layer vertex set be $R \cup B$.
We then define the \emph{lower layer} as a mirrored copy of the upper layer, where each vertex $v \in R \cup B$ has a corresponding copy $\hat{v}$ in the lower layer.

Our strategy involves two main phases: computing an approximate minimum-cost perfect matching of $\tilde{G}$ along with its associated dual weights, and utilizing these dual weights to derive the exact solution.
We begin by establishing the properties of an approximate solution and the associated dual weights.
\subparagraph*{1-optimality.}
Given a \emph{scaling factor} $\theta > 0$, let $\tilde{G}_{\theta} = (\tilde{V}, \tilde{E})$ denote the \emph{$\theta$-scaled graph} of $\tilde{G}$ with the edge-cost function $\tilde{c}_{\theta} : \tilde{E} \to \mathbb{Z}_{\ge 0}$, defined as $\tilde{c}_{\theta}(e) = \lceil \tilde{c}(e) / \theta \rceil$ for all $e \in \tilde{E}$.
Let $y(v)$ denote the \emph{dual weight} of a vertex $v \in \tilde{V}$.
A matching $\tilde{M}$ in $\tilde{G}_{\theta}$ and the associated dual weights $y(\cdot)$ are said to be \emph{1-feasible} if
\begin{align}
y(u) + y(v) &\le \tilde{c}_{\theta}(e) + 1 \quad \text{for all } e = (u,v) \in \tilde{E}, \label{ineq:1-feas-edge} \\
y(u) + y(v) &= \tilde{c}_{\theta}(e) \quad \text{for all } e = (u,v) \in \tilde{M}. \label{eq:1-feas-matching}
\end{align}
We refer to (\ref{ineq:1-feas-edge}) and (\ref{eq:1-feas-matching}) as the \emph{1-feasibility conditions}.
An edge $e = (u, v) \in \tilde{E} \setminus \tilde{M}$ is called \emph{admissible} if $y(u) + y(v) = \tilde{c}_{\theta}(e) + 1$.
The \emph{admissible graph} of $\tilde{G}_{\theta}$ is the subgraph of $\tilde{G}_{\theta}$ consisting of all admissible edges and the edges in $\tilde{M}$.
If a matching $\tilde{M}$ and dual weights $y(\cdot)$ satisfy the 1-feasibility conditions and $\tilde{M}$ is a perfect matching, we call the pair $(\tilde{M}, y)$ \emph{1-optimal}.
Sharathkumar and Agarwal \cite{DBLP:conf/soda/SharathkumarA12} showed that a 1-optimal matching of $\tilde{G}_{\theta}$ provides an approximate minimum-cost perfect matching of $\tilde{G}$ with an additive error $\varepsilon$ of at most $3n\theta$.

We now outline our algorithm for computing an exact minimum-cost perfect matching of $\tilde{G}$.
First, we compute a 1-optimal matching $\tilde{M}_1^*$ and associated dual weights $y(\cdot)$ of $\tilde{G}_{\theta}$ using a scaling factor of $\theta = 1 / (n \Delta^{33})$, following the procedure described in Section \ref{sec:scaling-algo-for-1-opt}.
In Section~\ref{sec:exact-algo-for-mcpm:candidate-construction}, we utilize these dual weights to identify a set of \emph{eligible edges} $E_{\textsf{elig}}$ of the upper layer of $\tilde G$ and exploit their geometric properties to construct a set $P$ of non-crossing line segments having endpoints at $B\cup R$.
Using $P$, we systematically extract a sparse candidate subgraph $\tilde{G}_{\textsf{cand}}$ of size $O(n)$ that is guaranteed to contain a minimum-cost perfect matching of $\tilde{G}$.
Finally, Section \ref{sec:exact-algo-for-mcpm:dnc-algo} presents a divide-and-conquer algorithm to compute an exact optimal solution in $\tilde{G}_{\textsf{cand}}$ by applying the planar separator theorem to the graph induced by $P$.

\subsection{Construction of the Candidate Subgraph}
\label{sec:exact-algo-for-mcpm:candidate-construction}
We construct a \emph{candidate subgraph} $\tilde{G}_{\textsf{cand}} = (\tilde{V}, \tilde{E}_{\textsf{cand}})$ of $\tilde{G}$ that is guaranteed to contain an optimal matching of $\tilde{G}$.
We begin by utilizing the 1-optimal matching $\tilde M_1^*$ and the 1-optimal dual weights of $\tilde{G}_{\theta}$ 
with a scaling factor of $\theta = 1 / (n \Delta^{33})$
to characterize the \emph{eligible edges} in the upper layer that can participate in an optimal solution.
This characterization reveals a crucial geometric property of these edges.
Using this property, we apply the technique of Sharathkumar \cite{DBLP:conf/compgeom/Sharathkumar13} to construct a planar skeleton (planar straight-line graph) $P$.
Finally, we use $P$ to generate the set of \emph{candidate edges} $\tilde{E}_{\textsf{cand}}$.

The overall construction comes from Sharathkumar~\cite{DBLP:conf/compgeom/Sharathkumar13}: It applies the approach to a Euclidean complete bipartite graph while we apply it to the upper layer of the prism graph $\tilde{G}$. As the upper layer is a Euclidean complete bipartite graph, the analysis immediately follows from~\cite{DBLP:conf/compgeom/Sharathkumar13}, but the detailed implementation should be tailored for our setting.

\subparagraph*{Eligible edge decomposition and its induced planar skeleton.}
We define the \emph{scaled dual weight} of a vertex $v \in \tilde{V}$ as $y_{\theta}(v) = \theta \cdot y(v)$.
The following lemma establishes a lower bound on the sum of scaled dual weights for the edges participating in an optimal matching of $\tilde{G}$.

\begin{lemma} \label{lem:ineq-for-opt-edge}
Let $y_{\theta}(\cdot)$ be the scaled dual weights associated with a $1$-optimal matching of $\tilde{G}_{\theta}$ using the scaling factor $\theta = 1 / (n\Delta^{33})$.
Let $\tilde{M}^*$ be a minimum-cost perfect matching of $\tilde{G}$.
Then, for every edge $(u,v) \in \tilde{M}^*$ in the upper layer, 
\begin{align}
y_{\theta}(u) + y_{\theta}(v) > \|u-v\| - \frac{2}{\Delta^{33}}. \label{ineq:elig-condition}
\end{align}
\end{lemma}
The proof is same as the argument in \cite{DBLP:conf/compgeom/Sharathkumar13} and is thus omitted.
We define the set of \emph{eligible edges}, denoted by $E_{\textsf{elig}}$, as the set of all edges in the upper layer of $\tilde G$ satisfying the condition~(\ref{ineq:elig-condition}). 
Lemma \ref{lem:ineq-for-opt-edge} implies that for any optimal matching $\tilde{M}^*$ of $\tilde{G}$, all edges of $\tilde{M}^*$ that belong to the upper layer are contained in $E_{\textsf{elig}}$.
Therefore, even if we remove all non-eligible edges in the upper layer from the prism graph, the optimal solution remains the same. 

The following lemma shows that if no four vertices of $R\cup B$ are collinear,
then the graph obtained from the upper layer of the prism graph by removing all non-eligible edges is planar.
Thus our main focus is to handle collinear points in the following. 

\begin{lemma}[\cite{DBLP:conf/compgeom/Sharathkumar13}] \label{lem:ineq-for-crossing-seg}
Let $Q \subset [\Delta]^2$ be a set of points.
Let $a,b,c,d \in Q$ be four distinct points that are not collinear.
If the line segments $\overline{ab}$ and $\overline{cd}$ intersect, then
\[
\|a - d\| + \|b - c\| + \frac{1}{\Delta^{32}} < \|a - b\| + \|c - d\|.
\]
Therefore, for any two distinct edges $(a, b)$ and $(c, d)$ in $E_{\textsf{elig}}$ such that their four endpoints are not collinear, 
the line segments $\overline{ab}$ and $\overline{cd}$ do not intersect.
\end{lemma}

We show that the non-crossing property enables us to compute an \emph{eligible edge decomposition} $P= \{(p_1, q_1), \ldots, (p_k, q_k)\}$, which is a set of non-crossing line segments, such that for every edge of $E_{\textsf{elig}}$,
it is contained in a segment of $P$. 
Note that $P$ induces a plane graph whose vertices come from $R\cup B$.

Since the definition of eligible edges in our setting is identical to that in Sharathkumar~\cite{DBLP:conf/compgeom/Sharathkumar13}, we can adopt their algorithm for computing the eligible edge decomposition. However, a minor technical hurdle arises in the initialization phase.
It uses an approximate \emph{perfect} matching $M$ obtained from the first phase, but in our setting, an approximate matching we have is not necessarily perfect if we look at the upper layer only. 
More specifically, it first computes all point-line pairs $(p,\ell)$ with $p\in R\cup B$ such that $\ell$ contains an eligible edge incident to $p$. The number of such pairs is $O(n)$. 
For a point $p\in R\cup B$, one can compute all eligible edges incident to it in $O(\polylog n)$ time per edge using a dynamic weighted nearest neighbor data structure.
However, it is possible that a single line $\ell$ contains many different eligible edges incident to $p$. Then a single line can be discovered  multiple times, which leads to quadratic running time in the worst case.
To avoid this, as initialization, it computes just one line $\ell$ for each point $p$ in $R\cup B$ containing an eligible edge incident to $p$ in $O(n\polylog n)$ time in total using $M$. Later, it considers each point $p$ in $B\cup R$ one by one and computes all pairs $(p,\ell)$ such that $\ell$ is not found in the initialization. Even in this case, a single line $\ell$ can be discovered multiple times, but the total number of pairs $(p,\ell)$ found in this way is $O(n)$, which leads to near-linear running time.
The perfect matching $M$ can be used in the initialization; since an edge $(u,v)$ of $M$ is eligible, its extension contains an eligible edge incident to $u$ (and $v$). But we can perform this initialization using a dynamic weighted nearest-neighbor data structure, and this does not increase the overall running time. The full details of this modified construction are deferred to Appendix~\ref{sec:eligible}.

\subparagraph*{Construction of the candidate subgraph.}
We now construct the candidate edge set $\tilde{E}_{\textsf{cand}}$ of size $O(n)$ using the eligible edge decomposition $P$. As mentioned earlier, if no four points of $R \cup B$ are collinear, we can simply construct $\tilde{E}_{\textsf{cand}}$ by adding all eligible edges (the edges of $P$), their mirrored edges in the lower layer, and all link edges between the layers. However, in the general case, the number of eligible edges can be $\Theta(n^2)$, making it impossible to include all of them.

Instead, we select a restricted subset of eligible edges as follows. For each pair $(p, q) \in P$, let $I$ be the set of points in $R \cup B$ lying on the open segment between $p$ and $q$. For a minimum-cost perfect matching $\tilde{M}^*$ of $\tilde G$, any edge incident to a vertex $v \in I$ must have its other endpoint in $I \cup \{ \hat{v}, p, q \}$. Consequently, for the endpoints $p$ and $q$, there are four possible configurations based on whether $p$ and $q$ are matched with points in $I$. 
For each configuration, we can identify the specific edges of $\tilde{M}^*$ incident to points in $I$ in $O(|I| \log |I|)$ time without computing the entire matching. More specifically, let $I' \subseteq \{p, q\}$ be the set of points matched with elements in $I$ under $\tilde{M}^*$. By solving the local matching problem for $I\cup I'$, we extract the necessary edges to be included in $\tilde{E}_{\textsf{cand}}$, ensuring the total size remains $O(n)$.
For the local matching problem,
since the vertices of $I\cup I'$ are collinear, the problem is restricted to one dimension. 
This problem reduces to the minimum-cost matching problem on $I\cup I'$, where leaving a vertex unmatched incurs a penalty of $\mu'(\cdot)$, where $\mu'(v)=\mu(v)$ for a vertex $v\in I$ and  $\mu'(v)=\infty$ for a vertex $v\in I'$.
Therefore, the local matching problem can be solved in $O(|I|\log |I|)$ time by Theorem~\ref{thm:matching-penalty}.
We do this for all the four configurations, and then take the union of the resulting edges.
Then $\tilde{E}_{\textsf{cand}}$ is defined as 
the union of all such edges in $P$, their mirrored edges in the lower layer, and all the link edges.
Then the size of $\tilde{E}_{\textsf{cand}}$ is $O(n)$, and it can be computed in $O(n\log n)$ time.

\medskip 
We show that $\tilde{E}_{\textsf{cand}}$ contains a minimum-cost perfect matching of $\tilde{G}$. For this, we use the following technical lemma justifying that focusing on the upper layer during the construction of $\tilde{E}_{\textsf{cand}}$ is sufficient.
We say a matching $M$ is \emph{symmetric} if
an edge $(u,v)$ in the upper layer is contained in $M$ if and only if its mirrored edge $(\hat{u}, \hat{v})$ is contained in $M$.

\begin{lemma}
\label{lem:symmetric-solution}
Let $\tilde{H}$ be a subgraph of $\tilde{G}$ induced by  edges in the upper layer, their mirrored edges in the lower layer, and all link edges. There exists a minimum-cost perfect matching $\tilde{M}^*$ of $\tilde{H}$ that is symmetric.
\end{lemma}
\begin{proof}
Let $\tilde{M}^*$ be an arbitrary minimum-cost perfect matching of $\tilde{H}$. Let $m$ be the number of edges in $\tilde{M}^*$ that belong to the upper layer. Each such edge $(u, v)$ covers one vertex in $R$ and one in $B$. Because $\tilde{M}^*$ is a perfect matching, any vertex in the upper layer not matched with a vertex in the upper layer must be matched to its corresponding copy in the lower layer via a link edge. 
This implies that exactly $m$ vertices in the lower layer  are not covered by link edges. To satisfy the perfect matching requirement, these $m$ vertices must be matched to each other using $m/2$ edges in the lower layer. 

Now, consider the cost of these $m/2$ edges. By construction, the cost of any edge in the lower layer is zero. Therefore, we can replace the existing lower-layer edges in $\tilde{M}^*$ with the ``mirrored'' edges of the upper-layer matching without changing the total cost of the matching. This substitution results in a new perfect matching $M'$ that remains optimal and satisfies the symmetry condition: $(u, v)$ is in $M'$ if and only if $(\hat{u}, \hat{v})$ is in $M'$.
Therefore, the lemma holds.
\end{proof}

\begin{lemma}
\label{lem:candidate-correctness}
There exists a minimum-cost perfect matching $\tilde{M}^*$ of $\tilde{G}$ such that $\tilde{M}^* \subseteq \tilde{E}_{\textsf{cand}}$.
\end{lemma}
\begin{proof}
Let $\tilde{M}^*$ be a minimum-cost perfect matching of $\tilde{G}$ such that (i) all edges of $\tilde{M}^*$ belonging to the upper layer are contained in $E_{\textsf{elig}}$, and (ii) $\tilde M^*$ is symmetric, which always exists by Lemma~\ref{lem:ineq-for-opt-edge} and by Lemma~\ref{lem:symmetric-solution}.

To show that $\tilde{M}^* \subseteq \tilde{E}_{\textsf{cand}}$, we consider the eligible edge decomposition $P$. Every edge $(u,v) \in \tilde{M}^*$ in the upper layer must be contained within some segment $\overline{p q} \in P$. For each such segment, let $I$ be the set of points of $B\cup R$ lying in the open interval between $p$ and $q$. By the construction of $P$, any edge in $\tilde{M}^*$ incident to a vertex in $I$ must have its other endpoint in $I \cup \{\hat{v}, p, q\}$. 
This local constraint implies that the matching within each segment $\overline{p q}$ can be optimized independently based on the four possible configurations of the endpoints $p$ and $q$. Since $\tilde{E}_{\textsf{cand}}$ is constructed by taking the union of the optimal local matchings for all four configurations across all segments in $P$, it follows that the edges of $\tilde{M}^*$ are included in $\tilde{E}_{\textsf{cand}}$. Thus, the subgraph of $\tilde G$ induced by $\tilde{E}_{\textsf{cand}}$ contains a minimum-cost perfect matching of $\tilde{G}$.
\end{proof}

Therefore, the candidate edge set $\tilde{E}_{\textsf{cand}}$ can be computed in $\tilde{O}(n)$ time in total, and the problem reduces to computing a minimum-cost perfect matching on the subgraph $\tilde{G}_{\textsf{cand}}$ of $\tilde G$ induced by $\tilde{E}_{\textsf{cand}}$. For an illustration of $\cand$, see Figure~\ref{fig:d-n-c-algo}.

It is worth noting the structural distinction between our construction and the one-to-one setting addressed by Sharathkumar~\cite{DBLP:conf/compgeom/Sharathkumar13}. In the one-to-one case, for each segment $(p,q) \in P$, it is sufficient to consider at most \emph{two} configurations of the endpoints $p$ and $q$, which makes the resulting graph planar. This is because 
the difference between the numbers of points in $R$ and of points in $B$ restricts the possible configurations of $p$ and $q$. In the many-to-many setting, however, this difference does not inherently reduce the number of configurations. Consequently, we must consider the union of local optimal matchings for all \emph{four} configurations, which renders the resulting candidate graph $\tilde{G}_{\textsf{cand}}$ non-planar.

\subsection{Divide-and-Conquer Algorithm via Planar Separators}
\label{sec:exact-algo-for-mcpm:dnc-algo}
In this section, we present a divide-and-conquer algorithm to compute a minimum-cost perfect matching of the candidate graph $\tilde{G}_{\textsf{cand}} = (\tilde{V}, \tilde{E}_{\textsf{cand}})$. 
Our algorithm recursively partitions $\tilde{G}_{\textsf{cand}}$ into disjoint vertex sets by removing a subset of vertices.
We first compute the optimal matching for the subgraph induced by the disjoint sets.
Subsequently, in the conquer phase, we reintroduce the removed vertices and gradually update the matching to obtain a minimum-cost perfect matching of $\tilde{G}_{\textsf{cand}}$ using augmenting paths.

We begin by recalling the planar separator theorem, which serves as the foundation for our graph partitioning strategy.
Lipton and Tarjan \cite{lipton1979separator} established the following result for planar graphs.

\begin{theorem}[\cite{lipton1979separator}]
\label{thm:planar-separator}
Let $H$ be a planar graph with $n_H$ vertices.
The vertices of $H$ can be partitioned in $O(n_H)$ time into three disjoint sets $X, Y, S$ such that (i) no edge connects a vertex in $X$ with a vertex in $Y$, (ii) $|X|, |Y| \le 2n_H/3$, and (iii) $|S| \le 2\sqrt{2}\sqrt{n_H}$.
\end{theorem}
We call $S$ a \emph{balanced separator} of $H$.
While $\tilde{G}_{\textsf{cand}}$ is not necessarily planar, its inherent planar structure is derived from the eligible edge decomposition $P$.
Let $V(P)$ be the set of endpoints of the edges (segments) in $P$, and let $G(P) = (V(P), P)$ be the planar skeleton defined by $P$, that is, it is the planar graph induced by the edges of $P$.

\begin{figure}
    \centering
    \includegraphics[width=1\linewidth]{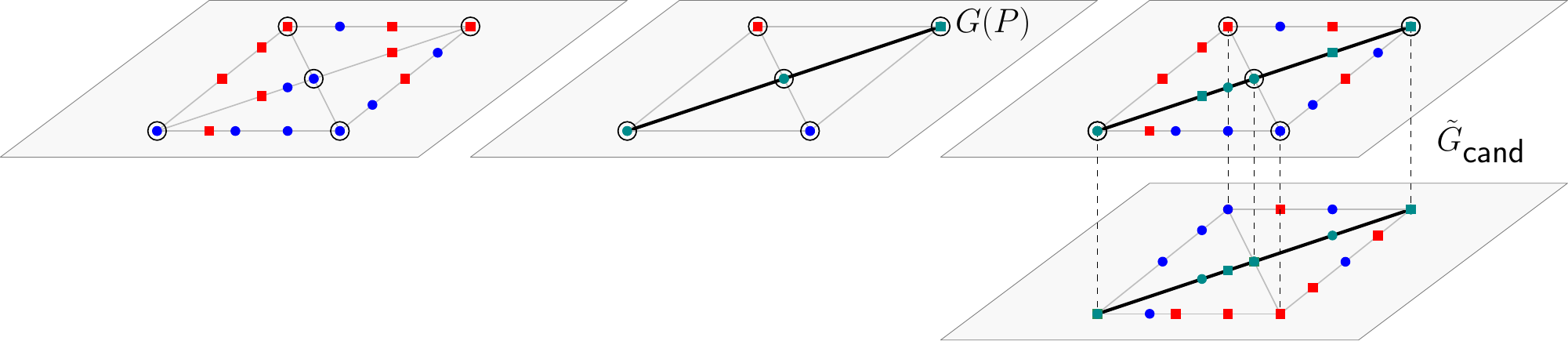}
    \caption{Illustration of the upper layer of $\cand$ (left), the planar skeleton $G(P)$ (middle), and the candidate subgraph $\tilde{G}_{\textsf{cand}}$ (right). Red squares and blue disks denote vertices in $R\cup \hat B$ and vertices in $B\cup \hat R$, respectively. In the upper layer of $\cand$, the segment endpoints represent skeleton vertices, while intermediate points on the segments represent interior vertices. The green vertices in $G(P)$ constitute the balanced separator of $G(P)$, and those in $\tilde{G}_{\textsf{cand}}$ form the prism separator $\tilde{S}$.}
    \label{fig:d-n-c-algo}
\end{figure}

We apply the planar separator theorem to $G(P)$ to induce a separation of $\tilde{G}_{\textsf{cand}}$.
Let $X, Y, S$ be a partition of $V(P)$ obtained by applying Theorem \ref{thm:planar-separator} to $G(P)$, where $S$ is a balanced separator of $G(P)$.
Let $\hat X, \hat Y$ and $\hat S$ denote the mirrored sets of $X, Y$ and $S$, respectively, in the lower layer.
We make the following observation regarding the sizes of these sets.
\begin{observation}
\label{obs:balanced-boundaries}
Let $\hat V(P)$ be the mirrored set of $V(P)$ in the lower layer.
Then, $|X \cup \hat X|, |Y \cup \hat Y| \le 2/3 \cdot (|V(P)| + |\hat V(P)|)$, and $|S\cup \hat S| \le 2\sqrt{2}(|V(P)| + |\hat V(P)|)^{1/2}$.
\end{observation}

Using the partition $(X,Y,S)$ of $V(P)$, we 
partition the vertex set of $\cand$ into $\tilde X,\tilde Y$ and $\tilde S$ with $X\subseteq \tilde X, Y\subseteq \tilde Y$ and $S\subseteq \tilde S$ as follows. 
For this, it suffices to partition the vertices of $\cand$
contained in the open segments of $P$. We call such vertices \emph{interior vertices}.
For each interior vertex $v$ of $\cand$, we assign it to the set based on the endpoints of the segment of $P$ containing $v$. Specifically, if the segment of $P$ containing $v$ has an endpoint in $X$ (in $Y$), we put $v$ and $\hat v$ to $\tilde X$ (to $\tilde Y$). 
Otherwise, both endpoints of the segment are contained in $S$. In this case, we put $v$ and $\hat v$ to $\tilde S$.
Here, notice that no segment of $P$ has endpoints both in $X$ and $Y$ since $S$ separates $X$ and $Y$.

\begin{lemma}
\label{lem:separator-of-cand-subgraph}
No edge in $\tilde{E}_{\textsf{cand}}$ connects a vertex in $\tilde{X}$ to a vertex in $\tilde{Y}$.
\end{lemma}
\begin{proof}
We prove this by examining the three types of edges $e$ in $\tilde{E}_{\textsf{cand}}$.
In the first case that $e$ belongs to the upper layer, 
every eligible edge in the upper layer is contained in a segment $s$ in $P$. 
If an endpoint of $s$ is contained in $X$ (or in $Y$),  
both endpoints of $e$ are in $\tilde X$ (or $\tilde Y$), or one belongs to $\tilde{X}$ (or $\tilde Y$) and the other to $\tilde{S}$. 
If both endpoints of $s$ are contained in $S$,
both endpoints of $e$ are contained in $\tilde S$. 
Thus in any case, the lemma holds.

In the other cases, the lemma holds 
due to the symmetric construction of $\tilde{E}_{\textsf{cand}}$ following Lemma \ref{lem:symmetric-solution}.
Specifically, a vertex $v\in B\cup R$ and its copy $\hat v$ belong to the same set among $\tilde X, \tilde Y$ and $\tilde S$. Therefore, 
in the case that $e$ belongs to the lower layer, or $e$ is a link edge, the lemma holds immediately.
\end{proof}
We refer to $\tilde{S}$ as the \emph{prism separator} of $\tilde{G}_{\textsf{cand}}$.
We classify each vertex of the prism separator $\tilde{S}$ into two types.
We call a vertex of $S \cup \hat S$ a \emph{skeleton vertex} of $\tilde S$ and a vertex of $\tilde S \setminus (S\cup \hat S)$ an \emph{interior vertex} of $\tilde S$.
See Figure~\ref{fig:d-n-c-algo}.

\subparagraph*{Divide-and-Conquer Algorithm.}
We now present a detailed divide-and-conquer algorithm for computing a minimum-cost perfect matching of $\tilde{G}_{\textsf{cand}}$. The algorithm first computes the partition $(\tilde{X}, \tilde{Y}, \tilde{S})$ of the vertex set in linear time and recursively finds minimum-cost perfect matchings for the subgraphs induced by $\tilde{X}$ and $\tilde{Y}$. Note that the subgraph induced by $\tilde X$ (and $\tilde Y$) has a perfect matching due to its prism structure. 
Since no edges connect $\tilde{X}$ and $\tilde{Y}$, the union of their optimal matchings is a minimum-cost perfect matching for the subgraph induced by $\tilde{X} \cup \tilde{Y}$. To obtain the optimal matching for the entire graph, one might consider extending this matching by adding all vertices of the separator $\tilde{S}$ one by one using the following lemma.
Here, a \emph{minimum-cost maximum-cardinality matching} is defined as the matching with the minimum total cost among all matchings of maximum possible cardinality. 

\begin{lemma}[\cite{DBLP:journals/siamcomp/LiptonT80}]
\label{lem:augment-matching}
Let $H$ be a graph with $n_H$ vertices and $m_H$ edges associated with edge costs, and let $v$ be a vertex of $H$. Given a minimum-cost maximum-cardinality matching of $H - \{v\}$, a minimum-cost maximum-cardinality matching of $H$ can be computed in $O(m_H \log n_H)$ time.
\end{lemma}

However, a naive application of Lemma~\ref{lem:augment-matching} to all vertices in $\tilde{S}$ would take $O(|\tilde{S}| \cdot n \log n)$ time, which is inefficient since $|\tilde{S}|$ can be as large as $\Theta(n)$. To address this, we observe that $\tilde{S}$ consists of $O(\sqrt{n})$ \emph{skeleton vertices} and a potentially large number of \emph{interior vertices}. Our strategy is to handle the skeleton vertices using Lemma~\ref{lem:augment-matching}, while the interior vertices are handled efficiently by exploiting their 1D nature and pre-computed optimal matchings.

\medskip 
The algorithm processes the subgraph $\tilde{G}'$ of  $\cand$ in each recursive step as follows. Using the subgraph of $G(P)$ induced by $V(P)\cap V(\tilde G')$,
we compute a balanced separator $S$ and a partition $(\tilde X,\tilde Y,\tilde S)$ of $\tilde G'$.
As mentioned earlier, we compute a minimum-cost perfect matching of $\tilde X\cup \tilde Y$ in linear time. Now consider the internal vertices of $\tilde S$.
We decompose them into subsets such that each subset consists of the internal vertices lying on the same segment of $P$. 
For each subset $W$, we retrieve its optimal matching of $W\cup \hat W$ using the pre-computed 1D matching from Section~\ref{sec:exact-algo-for-mcpm:candidate-construction}, where $\hat W$ is the copy of $W$ in the lower layer. 
An optimal matching of $W\cup\hat W$ can be obtained from a minimum-cost matching of $W$ with penalties, which can  be found in near-linear time.
Since no edge of $\tilde G'$ connects two vertices from different subsets, the union $M'$ of the all these matchings is a minimum-cost maximum-cardinality matching of $\tilde X\cup\tilde Y \cup \{ \text{internal vertices of }\tilde S\}$. Finally, we reintroduce the $O(\sqrt{n'})$  skeleton vertices of $\tilde S$ into the matching using Lemma~\ref{lem:augment-matching}, where $n'=|V(\tilde G')\cap V(P)|$. We iteratively apply Lemma~\ref{lem:augment-matching} for each skeleton vertex of $\tilde S$ to update $\tilde{M}'$. In this way, each recursion step takes $\tilde{O}(|V(\tilde G')|\sqrt {n'})$ time in total. 

In the base case that the number of (skeleton) vertices of $V(P)\cap V(\tilde G')$ is $O(1)$, 
we compute a minimum-cost maximum-cardinality matching of 
the (non-skeleton) vertices of $V(\tilde G')\setminus V(P)$ using their 1D nature in $\tilde O(|V(\tilde G')|)$ time.
Then we reintroduce the $O(1)$ skeleton vertices with Lemma~\ref{lem:augment-matching}. This takes $\tilde O(|V(\tilde G')|)$ time in total. 

\subparagraph*{Complexity Analysis.}
To bound the overall time complexity, let $\eta = O(n)$ be the number of vertices of $V(P)$. At depth $d$ of the recursion, let $H_{d,1}, H_{d,2}, \dots$ denote the disjoint subgraphs, where each $H_{d,i}$ contains $m_{d,i}$ edges and $\eta_{d,i}$ skeleton vertices along with their copies. By Theorem~\ref{thm:planar-separator} and Observation~\ref{obs:balanced-boundaries}, the number of skeleton vertices shrinks by a factor of at least $2/3$ per level, yielding $\eta_{d,i} \le (2/3)^d \eta$. Furthermore, since the subgraphs at any depth are edge-disjoint, $\sum_i m_{d,i} \le |\tilde{E}_{\textsf{cand}}| = O(n)$.

In the conquer phase, merging the sub-solutions for $H_{d,i}$ requires $O(\sqrt{\eta_{d,i}})$ augmentations. Since each augmentation takes $O(m_{d,i} \log n)$ time by Lemma~\ref{lem:augment-matching}, the total time for depth $d$ is:
\begin{align*}
\sum_i O\left(\sqrt{\eta_{d,i}} \cdot m_{d,i} \log n\right) 
\le O\left(\sqrt{(2/3)^d \eta} \log n\right) \sum_i m_{d,i} 
\le O\left(\left(\sqrt{2/3}\right)^d \cdot n^{1.5} \log n\right).
\end{align*}
Because $\sqrt{2/3} < 1$, the time for each level decreases geometrically. The total running time of the divide-and-conquer algorithm is thus dominated by the root level, which is $O(n^{1.5} \log n)$. 

\subparagraph*{Overall complexity and proof of Theorem~\ref{thm:main}.}
We compute an approximate matching and its dual weights for the scaled graph with a suitable scaling factor in  $\tilde{O}(n^{1.5} \log \Delta)$ time. 
The candidate graph can be computed in $O(n\log n)$ time, and the divide-and-conquer algorithm takes $O(n^{1.5}\log n)$ time. Therefore, the total time complexity of our algorithm remains $\tilde{O}(n^{1.5} \log \Delta)$.

\section{Conclusion}
In this paper, we presented an exact $\tilde{O}(n^{1.5} \log \Delta)$-time algorithm for the many-to-many matching problem on planar point sets with integer coordinates.
Our approach relies on a reduction to perfect matching on a prism graph, a scaling framework, and a planar separator-based divide-and-conquer strategy.

Extending this framework to non-integer coordinates or higher dimensions ($d \ge 3$) poses fundamental challenges. 
In particular, our framework relies on Lemma~\ref{lem:ineq-for-crossing-seg}. In continuous domains, the additive error in the inequality in Lemma~\ref{lem:ineq-for-crossing-seg} is no longer bounded, and thus it seems hard to apply our framework to continuous domains. 
For higher dimensions, geometric graphs lack planarity, preventing the use of a balanced vertex separator of $O(\sqrt{n})$ size.
Moreover, in higher dimensions, maintaining the dynamic weighted nearest-neighbor data structures required in the scaling step becomes substantially more difficult.
We leave the problem of designing a subquadratic-time algorithm for many-to-many matching on higher dimensions or continuous domains as an open problem.



\bibliography{main-bib}

\newpage

\appendix

\section{Algorithm for Computing an Eligible Edge Decomposition}
\label{sec:eligible}
Let $H$ be the upper layer of $\tilde G$, which is a geometric complete bipartite graph. 
Recall that an edge $(u,v)$ of $H$ satisfying the following is called an eligible edge:  
\begin{align*}
y_{\theta}(u) + y_{\theta}(v) > \|u- v\| - \frac{2}{\Delta^{33}}. 
\end{align*}

The overall construction of the candidate subgraph proceeds in two main steps:
(i) we generate a set $\mathcal P$ of point-line pairs to identify the lines supporting the eligible edges, and 
(ii) we utilize $\mathcal P$ to construct the eligible edge decomposition $P$.
Note that the definition of eligible edges here is identical to the one in~\cite{DBLP:conf/compgeom/Sharathkumar13}. Therefore, the combinatorial properties used in~\cite{DBLP:conf/compgeom/Sharathkumar13} still hold in this case while we cannot use an approximate perfect matching obtained from the first phase to construct $P$.

\subparagraph*{Computation of point-line pairs.}
Let $v \in R\cup B$ be a point and $\ell$ be a line. We say that $v$ is \emph{active} with respect to $\ell$ if there exists a point $u \in R\cup B$ such that the edge $(v, u)$ is in $E_{\textsf{elig}}$ and the segment $\overline{vu}$ is contained in $\ell$. Our goal is to compute the set $\mathcal P$ of all point-line pairs $(v, \ell)$ where $v$ is active with respect to $\ell$.

For this, we maintain a dynamic set $W$ of vertices, initialized as $R \cup B$, with weights $y_\theta(\cdot)$, and maintain a dynamic weighted nearest neighbor data structure~\cite{DBLP:journals/dcg/KaplanMRSS20} on $W$. Given a query vertex $v \in R \cup B$, we find $u = \arg\min_{w \in W} (\|v-w\| - y_\theta(w))$. Due to the feasibility of the dual weights, $(v, u) \in E_{\textsf{elig}}$ if there exists any eligible edge incident to $v$. To retrieve additional eligible edges, one could remove $u$ from $W$ and repeat the query. However, even if $|\mathcal{P}|$ is small, a naive approach of retrieving every eligible edge requires $\Omega(n^2)$ time because many edges incident to $v$ may lie on the same line $\ell$, causing $(v, \ell)$ to be identified redundantly.

To avoid this, we follow the observation by Sharathkumar \cite{DBLP:conf/compgeom/Sharathkumar13} that the number of eligible edges incident to \emph{external vertices} is $O(n)$. Here, a vertex $b \in R \cup B$ is \emph{external} if it is active with respect to at least two distinct lines. The algorithm proceeds as follows: (i) For each $v \in R \cup B$, compute an arbitrary line $\ell(v)$ such that $v$ is active with respect to $\ell(v)$, and (ii) for each $v \in R \cup B$, compute all eligible edges $(u, v)$ such that $\ell(u) \neq \ell(v)$. 
The total number of eligible edges found in Step (ii) is $O(n)$ due to the property of external vertices.
Step (i) is performed using the eligible oracle in $O(n \polylog n)$ time.\footnote{In~\cite{DBLP:conf/compgeom/Sharathkumar13}, they perform this in $O(n \log n)$ time using a perfect matching from the first phase. Since our first phase does not yield a perfect matching, we use the oracle instead. This remains efficient as it is subsumed by the time for Step (ii).} Step (ii) can be performed as follows. We partition $R \cup B$ into subsets $X_\ell$ such that all vertices in each $X_\ell$ share the same line $\ell(\cdot)$. For each subset $X_\ell$, we perform Step (ii) by temporarily removing all points of $X_\ell$ from $W$ and repeatedly querying the oracle for each $v \in X_\ell$ until it returns an edge that is no longer eligible. Since we only search for edges $(v, u)$ where $u \notin X_\ell$ (meaning $\ell(u) \neq \ell(v)$), the number of such edges is $O(n)$, and the total time for all subsets is $O(n \polylog n)$. Thus, $\mathcal{P}$ can be computed in $O(n \polylog n)$ time.

\subparagraph*{Construction of an eligible edge decomposition.}
Given the set of point-line pairs $\mathcal P$, we obtain the eligible edge decomposition $P$. 
Let $\mathcal L$ be the set of lines appearing in point-line pairs in $\mathcal P$.
Since the arrangement of lines of $\mathcal P$ 
has complexity $\Omega(n^2)$, the strategy of~\cite{DBLP:conf/compgeom/Sharathkumar13} is the following: For each line $\ell\in \mathcal L$, compute the union $U(\ell)$ of all eligible edges (segments) contained in $\ell$. Here, $U(\ell)$ consists of several disjoint segments on $\ell$. Then the total complexity of $U(\cdot)$ for all lines in $\mathcal L$ is $O(n)$, and for any two lines $\ell,\ell'$ in $\mathcal L$, we have no two segments from $U(\ell)$ and from $U(\ell')$ cross by Lemma~\ref{lem:ineq-for-crossing-seg}. Therefore, this yields the eligible edge decomposition with the desired properties.

Thus it suffices to show how to compute $U(\ell)$ for a fixed line $\ell\in \mathcal L$. First, 
we partition $\ell$ into segments with respect to the points $v$ with $(v,\ell)\in\mathcal P$. 
For each segment $s$ on $\ell$, we construct a complete bipartite graph $G_s$ with vertex set $(R\cap s, B\cap s)$ where the cost of an edge $(u, v)$ is defined as $\|u - v\| - y_{\theta}(u) - y_{\theta}(v)$. 
Let $F_s$ be the forest obtained from a minimum spanning tree of $G_s$ by removing all edges that do not satisfy the eligibility condition (\ref{ineq:elig-condition}). Since all vertices of $G_s$ lie on the line and the edge costs are defined with respect to the Euclidean distance, we can compute its minimum spanning tree in $O(n\log n)$ time.
By construction, each tree of $F_s$ spans a geometric interval of $s$, and it is contained in the same component (segment) of $U(s)$. But it is possible that two geometric intervals of $s$ spanned by two distinct trees of $F_s$ might be contained in the same component of $U(s)$. To avoid this, we merge the trees in $F_s$ if the geometric intervals spanned by their edges overlap. Since the number of subtrees is $O(n)$, this can be done in $O(n\log n)$ time in total.

In this way, an eligible edge decomposition $P$ of the upper layer of $\tilde G$ can be computed in $O(n\polylog n)$ time in total.
For details, refer to~\cite{DBLP:conf/compgeom/Sharathkumar13}.

\section{Scaling Algorithm for 1-Optimal Matching}
\label{sec:scaling-algo-for-1-opt}
In this section, we present an algorithm to compute a 1-optimal matching $\tilde{M}^*_1$ and the corresponding dual weights $y(\cdot)$ of the $\theta$-scaled graph $\tilde{G}_{\theta}$ of $\tilde{G}$, given an error parameter $\varepsilon > 0$ and a scaling factor $\theta$ satisfying $\theta \le \varepsilon / (3n)$.
Our approach is essentially the same as that of Bandyapadhyay et al.~\cite{DBLP:conf/isaac/BandyapadhyayMS21}, though we have slightly modified the description of the algorithm for consistency.
We establish the following theorem, which is the main result of this section.
\begin{theorem}\label{thm:sub}
For any $\varepsilon > 0$ and a scaling factor $\theta$ satisfying $\theta \leq \varepsilon / (3n)$, a 1-optimal matching $\tilde{M}^*_1$ and the corresponding dual weights $y(\cdot)$ of $\tilde{G}_{\theta}$ can be computed in $\tilde O\!\left(n^{1.5}\, \log {(\Delta / \varepsilon)}\right)$ time.
\end{theorem}
Recall that $\tilde{G}_{\theta}$ denotes the graph $\tilde{G}$ equipped with the scaled edge costs $\tilde{c}_{\theta}(e) = \lceil \tilde{c}(e)/\theta \rceil$ for all $e \in \tilde{E}$.
We detail the scaling algorithm for computing a 1-optimal matching in Section~\ref{sec:scaling-algo-for-1-opt:algorithm}.
Subsequently, in Section~\ref{sec:scaling-algo-for-1-opt:data-structures}, we describe the data structure $\mathcal{D}$ required to support the search procedures efficiently and conclude with the proof of Theorem \ref{thm:sub}.

\subparagraph*{Notations for the prism graph.}
Let $R_0 = R$ and $B_0 = B$.
We introduce two sets of vertices: $R_1$, which contains a copy $\hat{b}$ for each vertex $b \in B_0$, and $B_1$, which contains a copy $\hat{r}$ for each vertex $r \in R_0$.
The vertex set of $\tilde{G}$ is defined as $\tilde{V} = \tilde{R} \cup \tilde{B}$, where $\tilde{R} = R_0 \cup R_1$ and $\tilde{B} = B_0 \cup B_1$.
The edge set $\tilde{E}$ is defined as the union of four disjoint sets, $\tilde{E} = \bigcup_{i = 0}^{3} E_i$, defined as follows:
\begin{align*}
E_0 &= R_0 \times B_0, \text{ } E_1 = R_1 \times B_1, \\
E_2 &= \{(r, \hat{r}) \mid r \in R_0, \, \hat{r} \in B_1 \text{ is the copy of } r\}, \\
E_3 &= \{(\hat{b}, b) \mid b \in B_0, \, \hat{b} \in R_1 \text{ is the copy of } b\}.
\end{align*}
We refer to the edges in $E_2 \cup E_3$ as \emph{link edges}.
For each edge $e = (r, b) \in E_0$, let $e'$ be the edge $(\hat{b}, \hat{r}) \in E_1$, where $\hat{b}$ and $\hat{r}$ are the copies of $b$ and $r$, respectively.
We refer to $e'$ as the \emph{copied edge} of $e$.
Similarly, for any subset $E' \subseteq E_0$, we define the \emph{copied set} of $E'$ as the set of copied edges corresponding to all edges in $E'$.
\subsection{Scaling Algorithm}
\label{sec:scaling-algo-for-1-opt:algorithm}
We describe the algorithm to compute a 1-optimal matching $\tilde{M}^*_1$ and the associated dual weights $y(\cdot)$ of $\tilde{G}_{\theta}$.
The total running time is $\tilde O \bigl(n^{1.5} \Phi(n) \log(n\Delta / \varepsilon)\bigr)$, where $\Phi(n)$ denotes the query and update time of a dynamic weighted nearest neighbor data structure \cite{DBLP:conf/soda/SharathkumarA12}.

\subparagraph*{Overall procedure.}
We initialize the scaling factor $\theta$ based on the diameter of the graph and set $y(v) = 0$ for all $v \in \tilde{V}$.
The algorithm proceeds in phases, iterating until $\theta \le \varepsilon / (3n)$.
In each phase, we update $\theta \leftarrow \theta/2$ and adjust the dual weights as $y(v) \leftarrow 2y(v) - 1$ for all $v \in \tilde{V}$.
With the updated parameters, we invoke the \textsc{1-OptimalMatch} procedure to compute a 1-optimal matching for $\tilde{G}_{\theta}$.

The \textsc{1-OptimalMatch} procedure initializes $\tilde{M} = \emptyset$ and iteratively expands it to a perfect matching while maintaining the 1-feasibility of $y(\cdot)$.
This is achieved by repeating the following three steps:
\begin{enumerate}
\item \textbf{Hungarian search.} Perform a Hungarian search to adjust dual weights and uncover new admissible edges.
\item \textbf{Depth-first search.} Identify a maximal set $\mathcal{P}$ of vertex-disjoint augmenting paths in the admissible graph using a depth-first search.
\item \textbf{Augmentation.} Augment $\tilde{M}$ along the augmenting paths in $\mathcal{P}$.
\end{enumerate}
We detail the search procedures below.

\subparagraph*{Hungarian search.}
The primary objective of the Hungarian search is to adjust the dual weights to create new admissible edges until at least one augmenting path is found.
To achieve this, the search maintains an \emph{alternating forest} $\mathcal{F}$, which is defined as a collection of trees where each tree is rooted at a free vertex in $\tilde{B}$, and every path extending from a root is an alternating path with respect to $\tilde{M}$.
Let $\tilde{R}_{\mathcal{F}}$ and $\tilde{B}_{\mathcal{F}}$ denote the sets of vertices in $\tilde{R}$ and $\tilde{B}$ currently contained in the forest $\mathcal{F}$, respectively.
Initially, $\tilde{R}_{\mathcal{F}} = \emptyset$, and $\tilde{B}_{\mathcal{F}}$ contains all free vertices of $\tilde{B}$.

In each step, we compute the \emph{minimum slack} $\alpha$ over all edges connecting $\tilde{B}_{\mathcal{F}}$ to $\tilde{R} \setminus \tilde{R}_{\mathcal{F}}$ defined as 
\begin{align}
\alpha \;=\; \min_{b \in \tilde{B}_{\mathcal{F}},\, r \in \tilde{R} \setminus \tilde{R}_{\mathcal{F}}} \left\{ \tilde{c}_{\theta}(r, b) + 1 - y(r) - y(b) \right\}. \label{eq:alpha}
\end{align}
If $\alpha > 0$, we perform a dual update: $y(b) \leftarrow y(b) + \alpha$ for all $b \in \tilde{B}_{\mathcal{F}}$ and $y(r) \leftarrow y(r) - \alpha$ for all $r \in \tilde{R}_{\mathcal{F}}$.
This update maintains 1-feasibility while reducing the slack of edges between $\tilde{B}_{\mathcal{F}}$ and $\tilde{R} \setminus \tilde{R}_{\mathcal{F}}$ by exactly $\alpha$.
To avoid the explicit update of $O(n)$ dual variables in each step, we employ the implicit update technique introduced by Vaidya \cite{DBLP:journals/siamcomp/Vaidya89a}.
We introduce a \emph{global offset} $\omega$, initialized to $0$ at the beginning of the search.
Additionally, each vertex $v \in \tilde{V}$ is assigned a static weight $\sigma_v$.
Initially, we set $\sigma_v = y(v)$ for all $v \in \tilde{V}$.
When a vertex $r$ is added to $\tilde{R}_{\mathcal{F}}$, we set its static weight to $\sigma_r = y(r) + \omega$.
When a vertex $b$ is added to $\tilde{B}_{\mathcal{F}}$, we set its static weight to $\sigma_b = y(b) - \omega$.
When a dual update is required, i.e., $\alpha > 0$, we simply increment the global offset by updating $\omega \leftarrow \omega + \alpha$.
This mechanism leads to the following structural property of the dual weights.
\begin{observation}[\cite{DBLP:journals/siamcomp/Vaidya89a}]
\label{obs:implicit-dual}
At any point during the search, the dual weight $y(v)$ of a vertex in $v \in \tilde{V}$ is implicitly maintained as $y(v) = \sigma_v - \omega$ if $v \in \tilde{R}_{\mathcal{F}}$, $y(v) = \sigma_v + \omega$ if $v \in \tilde{B}_{\mathcal{F}}$, and $y(v) = \sigma_v$ otherwise.
\end{observation}
Based on Observation \ref{obs:implicit-dual}, the minimum slack $\alpha$ in \eqref{eq:alpha} can be equivalently rewritten using the static weights as 
\begin{align}
\alpha \;=\; \min_{b \in \tilde{B}_{\mathcal{F}},\, r \in \tilde{R} \setminus \tilde{R}_{\mathcal{F}}} \left\{ \tilde{c}_{\theta}(r, b) + 1 - \sigma_r - \sigma_b \right\} - \omega. \label{eq:alpha-implicit}
\end{align}
This allows us to implicitly track and update the dual weights in $O(1)$ time by merely updating $\omega$, effectively reducing the Hungarian search to the task of maintaining the bichromatic closest pair between $\tilde{B}_{\mathcal{F}}$ and $\tilde{R} \setminus \tilde{R}_{\mathcal{F}}$ with respect to the static distance function $\tilde{c}_{\theta}(r, b) + 1 - \sigma_r - \sigma_b$.
The search proceeds until $\alpha = 0$, at which point an admissible edge $(r, b)$ achieving the minimum in \eqref{eq:alpha-implicit} is identified.
If $r$ is free, an augmenting path has been found, and the Hungarian search terminates.
If $r$ is matched to a vertex $b' \in \tilde{B}$ in $\tilde{M}$, we extend $\mathcal{F}$ by adding $r$ to $\tilde{R}_{\mathcal{F}}$ and $b'$ to $\tilde{B}_{\mathcal{F}}$, and the search iterations continue.

\subparagraph*{Depth-first search.}
The depth-first search procedure identifies a maximal set $\mathcal{P}$ of vertex-disjoint augmenting paths.
We initialize $\mathcal{P} = \emptyset$ and mark all vertices as unvisited.
For each unvisited free vertex $b \in \tilde{B}$, we initiate an alternating path $\Gamma$ starting at $b$ to explore the admissible graph of $\tilde{G}_{\theta}$.

Let $u \in \tilde{B}$ denote the current endpoint of $\Gamma$, and let $U_{\tilde{R}} \subseteq \tilde{R}$ be the set of currently unvisited vertices in $\tilde{R}$.
To find a valid extension for $\Gamma$, we compute the \emph{local minimum slack} $\beta(u)$ with respect to $u$, defined as
\begin{align}
\beta(u) \;=\; \min_{v \in U_{\tilde{R}}} \left\{ \tilde{c}_{\theta}(v, u) + 1 - y(v) - y(u) \right\}. \label{eq:beta}
\end{align}
We evaluate $\beta(u)$ to determine the subsequent step in the search:
\begin{itemize}
\item \textbf{Case 1. $\beta(u) = 0$.} There exists an unvisited vertex $v \in U_{\tilde{R}}$ achieving (\ref{eq:beta}), meaning $(v, u)$ is an admissible edge. We mark $v$ as visited and extend $\Gamma$ by adding $(v, u)$.
If $v$ is free, $\Gamma$ constitutes a valid augmenting path, so we add $\Gamma$ to $\mathcal{P}$ and terminate the search initiated from $b$.
If $v$ is matched to some $u' \in \tilde{B}$ in $\tilde{M}$, we mark $u'$ as visited, further extend $\Gamma$ with $(v, u')$, and continue the search from $u'$.
    
\item \textbf{Case 2. $\beta(u) > 0$.} There is no admissible edge connecting $u$ to any unvisited vertex in $\tilde{R}$. We backtrack by removing $u$ and its matched vertex in $\tilde{R}$, if any, from $\Gamma$, and continue the search from the preceding vertex in $\tilde{B}$ on $\Gamma$.
\end{itemize}
This process repeats until $\Gamma$ becomes empty, at which point we proceed to the next unvisited free vertex in $\tilde{B}$.
Maintaining the visited status of vertices across the searches ensures that each vertex is explored at most once, thereby preventing redundant traversals and bounding the overall time complexity of the procedure.

The correctness of our scaling algorithm follows directly from the analysis by \cite{DBLP:conf/soda/SharathkumarA12}.
To bound the time complexity, let $\Phi(n)$ denote the maximum time required by the underlying data structure $\mathcal{D}$ to compute the global minimum slack $\alpha$ in (\ref{eq:alpha-implicit}) or the local minimum slack $\beta(u)$ in (\ref{eq:beta}).
For a fixed scaling factor $\theta$, the \textsc{1-OptimalMatch} procedure performs $O(\sqrt{n})$ iterations, which cumulatively require $O(n^{1.5})$ data structure operations, taking $O(n^{1.5} \Phi(n))$ time per scaling phase \cite{DBLP:conf/soda/SharathkumarA12}.
Since the scaling factor is repeatedly halved from an initial value proportional to the diameter $\Delta$ down to $\varepsilon / (3n)$, there are $O(\log(n \Delta / \varepsilon))$ phases in total \cite{DBLP:journals/siamcomp/GabowT89}.
Consequently, the overall running time is bounded by $O\bigl(n^{1.5} \Phi(n) \log(n\Delta / \varepsilon)\bigr)$.
In Section \ref{sec:scaling-algo-for-1-opt:data-structures}, we will show that $\mathcal{D}$ supports these queries and updates in $\Phi(n) = \text{polylog}(n)$ time, which yields a total running time of $\tilde O(n^{1.5} \log (\Delta / \varepsilon))$ and completes the proof of Theorem~\ref{thm:sub}.

\subsection{Data Structures}
\label{sec:scaling-algo-for-1-opt:data-structures}
In this section, we detail the data structures that support the queries for the Hungarian search and depth-first search in $\Phi(n) = \text{polylog}(n)$ time.
Following the approach of Bandyapadhyay et al.~\cite{DBLP:conf/isaac/BandyapadhyayMS21}, we exploit the distinct geometric and structural properties of the edge sets $E_0, E_1, E_2$, and $E_3$ by maintaining separate data structure components for each.

To handle the edges in $E_0$ efficiently, we must relate the scaled cost $\tilde{c}_{\theta}(\cdot, \cdot)$ to the original Euclidean distance.
The following lemma establishes that minimizing the scaled slack is equivalent to minimizing an additively weighted distance \cite{DBLP:conf/soda/SharathkumarA12}.

\begin{lemma}[\cite{DBLP:conf/soda/SharathkumarA12}]
\label{lem:scaled-argmin}
Let $R' \subseteq R_0$ and $B' \subseteq B_0$. 
If $(r^*, b^*) = \arg\min_{r \in R', b \in B'} \left\{ \|r - b\| - \theta \sigma_r - \theta \sigma_b \right\}$, then $(r^*, b^*)$ minimizes $\tilde{c}_{\theta}(r, b) + 1 - \sigma_r - \sigma_b$ among all pairs in $R' \times B'$.
\end{lemma}
\begin{proof}
By definition, the scaled edge cost is $\tilde{c}_{\theta}(r, b) = \lceil \|r - b\| / \theta \rceil$ for an edge $(r, b) \in E_0$.
Since $\sigma_r$ and $\sigma_b$ represent dual weights and global offsets, they are integers throughout the algorithm.
This allows us to bring them inside the ceiling function:
\begin{align*}
\min_{r \in R', b \in B'} \left\{ \tilde{c}_{\theta}(r, b) + 1 - \sigma_r - \sigma_b \right\} &= \min_{r \in R', b \in B'} \left\{ \left\lceil \frac{\|r - b\|}{\theta} \right\rceil - \sigma_r - \sigma_b \right\} + 1 \\
&= \min_{r \in R', b \in B'} \left\lceil \frac{\|r - b\| - \theta \sigma_r - \theta \sigma_b}{\theta} \right\rceil + 1.
\end{align*}
Because dividing by $\theta$, taking the ceiling, and adding a constant are all monotonic operations, the pair $(r^*, b^*)$ that achieves the minimum for $\|r - b\| - \theta \sigma_r - \theta \sigma_b$ also minimizes the value above.
\end{proof}

\subparagraph*{Data structure for the Hungarian search.}
For the Hungarian search, the data structure $\mathcal{D}$ consists of five components: $D$, $H_{R_1}$, $H_{B_1}$, $H_2$, and $H_3$.
Each component is designed to efficiently evaluate the static distance function for its respective edge set.
\begin{itemize}
\item \textbf{$D$ for $E_0$.} 
We employ the dynamic additively weighted nearest neighbor data structure of Kaplan et al. \cite{DBLP:journals/dcg/KaplanMRSS20} as $D$. 
It maintains the points $r \in R_0 \setminus \tilde{R}_{\mathcal{F}}$ with an additive weight of $-\theta \sigma_r$, and $b \in B_0 \cap \tilde{B}_{\mathcal{F}}$ with an additive weight of $-\theta \sigma_b$.
By Lemma \ref{lem:scaled-argmin}, a nearest neighbor query on $D$ directly yields the edge in $E_0$ that minimizes the slack.

\item \textbf{$H_{R_1}$ and $H_{B_1}$ for $E_1$.} 
Since all edges in $E_1 = R_1 \times B_1$ have zero cost, minimizing the slack $\tilde{c}_\theta(r, b) + 1 - \sigma_r - \sigma_b$ is equivalent to maximizing $\sigma_r + \sigma_b$. 
Thus, we maintain two max-heaps, $H_{R_1}$ and $H_{B_1}$, storing the keys $\sigma_r$ for $r \in R_1 \setminus \tilde{R}_{\mathcal{F}}$ and $\sigma_b$ for $b \in B_1 \cap \tilde{B}_{\mathcal{F}}$, respectively.

\item \textbf{$H_2$ and $H_3$ for $E_2 \cup E_3$.}
We maintain a min-heap $H_2$ for the link edges $(r, \hat{r}) \in E_2$ with keys $\tilde{c}_{\theta}(r, \hat{r}) + 1 - \sigma_r - \sigma_{\hat{r}}$, and a min-heap $H_3$ for the link edges $(\hat{b}, b) \in E_3$ with keys $\tilde{c}_{\theta}(\hat{b}, b) + 1 - \sigma_{\hat{b}} - \sigma_b$. 
Only the edges currently connecting $\tilde{B}_{\mathcal{F}}$ to $\tilde{R} \setminus \tilde{R}_{\mathcal{F}}$ are stored in these heaps.
\end{itemize}

\subparagraph*{Data structure for the depth-first search.}
For the depth-first search, we maintain a separate data structure, denoted as $\mathcal{D}'$.
It follows the same structure as $\mathcal{D}$ and consists of two components, $D'$ and $H'_{R_1}$, each corresponding to its counterpart in $\mathcal{D}$.
Specifically, instead of tracking vertices based on their membership in the alternating forest $\mathcal{F}$, $\mathcal{D}'$ dynamically maintains the partitions of currently visited and unvisited vertices to support the search of admissible edges.

We now explain how the search procedures are supported.
During the Hungarian search, we compute the global minimum slack $\alpha$ defined in (\ref{eq:alpha-implicit}) by retrieving the minimum candidate from each data structure in $\mathcal{D}$.
The global minimum among these four candidates, offset by $-\omega$, yields the exact value of $\alpha$.
During the depth-first search, we efficiently compute the local minimum slack $\beta(u)$ defined in (\ref{eq:beta}) for a given vertex $u \in \tilde{B}$ by evaluating its incident edges using $\mathcal{D}'$.
\begin{itemize}
\item \textbf{Case 1. $u \in B_0$.} 
The candidates for $u$ are either edges to $R_0$ or the link edge to its copy $\hat{u} \in R_1$. 
We perform exactly one query on $D'$ to find the minimum slack among edges to unvisited vertices in $R_0$. 
We also evaluate the slack of its specific link edge to $\hat{u}$, provided that $\hat{u}$ is unvisited. 
The value $\beta(u)$ is simply the minimum of these two values.

\item \textbf{Case 2. $u \in B_1$.}
The candidates for $u$ are either edges to $R_1$ or the link edge to its original vertex $v \in R_0$.
We perform exactly one query on the max-heap $H'_{R_1}$ to find an unvisited vertex in $R_1$ that minimizes the slack in $E_1$.
Similarly, we explicitly evaluate the slack of its link edge to $v$, provided that $v$ is unvisited.
The value $\beta(u)$ is the minimum of these two values.
\end{itemize}

We conclude by analyzing the operation time $\Phi(n)$ of the data structures $\mathcal{D}$ and $\mathcal{D}'$ to complete the proof of Theorem \ref{thm:sub}.
The components $D$ and $D'$ are implemented using the dynamic additively weighted nearest neighbor data structure by Kaplan et al. \cite{DBLP:journals/dcg/KaplanMRSS20}, which supports updates and queries in $\text{polylog}(n)$ time.
The remaining components are standard priority queues.
Updates and queries on these are performed in $O(\log n)$ time.
Since any query or update on $\mathcal{D}$ or $\mathcal{D}'$ during the search procedures involves a constant number of operations on these components, the operation time is dominated by $D$ and $D'$.
Consequently, we obtain $\Phi(n) = \polylog n$.

Substituting $\Phi(n) = \polylog n$ into the overall time complexity $O\bigl(n^{1.5} \Phi(n) \log(n\Delta / \varepsilon)\bigr)$ established in Section \ref{sec:scaling-algo-for-1-opt:algorithm}, the total running time of the scaling algorithm is bounded by $\tilde{O}\bigl(n^{1.5} \log(\Delta / \varepsilon)\bigr)$.
This is the claimed time complexity and completes the proof of Theorem~\ref{thm:sub}.

\end{document}